\documentclass{article}
\usepackage{graphicx} 
\usepackage[utf8]{inputenc} 
\usepackage[T1]{fontenc}    
\usepackage{hyperref}       
\usepackage{url}            
\usepackage{booktabs}       
\usepackage{amsfonts}       
\usepackage{nicefrac}       
\usepackage{microtype}      
\usepackage{xcolor}      
\usepackage{natbib}

\usepackage{epstopdf}
\usepackage{amsmath, amssymb}
\usepackage{amsfonts}              
\usepackage{amsthm}                
\usepackage{caption} 
\usepackage{subcaption}
\usepackage{graphicx}           
\usepackage{dsfont}
\usepackage{bbm}
\usepackage{ifsym}
\usepackage{color}
\usepackage{istgame} 
\usepackage{algorithm}
\usepackage{algorithmic}
\usepackage{accents}
\usepackage{mathtools}

\theoremstyle{plain}
\newtheorem{theorem}{Theorem}[section]
\newtheorem{proposition}[theorem]{Proposition}
\newtheorem{lemma}[theorem]{Lemma}

\theoremstyle{definition}
\newtheorem{definition}[theorem]{Definition}

\theoremstyle{remark}

\newcommand{\Rho}{\mathrm{P}}

\title{A Proof that Coarse Correlated Equilibrium Implies Nash Equilibrium in Two-Player Zero-Sum Games}
\author{Revan MacQueen}
\date{2023}

\begin{document}

\maketitle

\begin{abstract}
    We give a simple proof of the well-known result that the marginal strategies of a coarse correlated equilibrium form a Nash equilibrium in two-player zero-sum games. A corollary of this fact is that no-external-regret learning algorithms that converge to the set of coarse correlated equilibria will also converge to Nash equilibria in  two-player zero-sum games. We show an approximate version: that $\epsilon$-coarse correlated equilibria imply $2\epsilon$-Nash equilibria.
\end{abstract}

\section{Background}
\begin{flushleft}

It is well-known that the marginal strategies of a coarse correlated equilibrium (CCE) are Nash equilibria in two-player zero-sum games. We present a simple proof of this fact, which may be useful to others. 

\medskip

This fact is useful, since a no-external-regret learning algorithm's average strategy profile will approach the set of CCE \citep{Greenwald-2003}, and so will compute an approximate Nash equilibrium in two-player zero-sum games. Hence, the marginal strategies of a CCE in two-player zero-sum games have  the desirable properties of Nash equilibrium, including bounded exploitability \citep{Fudenberg-1991}.

\begin{definition}
    A normal form game $G$ is a 3 tuple $G = (N, \Rho, u)$ where $N$ is a set of players, $\Rho = \bigtimes_{i \in N} \Rho_i$ is a joint pure strategy space where $\Rho_i$ is a set of \emph{pure strategies} for player $i$. $u = (u_i)_{i \in N}$ is a set of utility functions where $u_i : \Rho \to \mathbb{R}$.
\end{definition}

Pure strategies are deterministic choices of actions in the game. In a normal-form game, each agent simultaneously  chooses their pure strategy $\rho_i \in \Rho_i$; we  call a joint selection of pure strategies $\rho \in \Rho$ a \emph{pure strategy profile}. The pure strategy profile determines the payoff (or utility) to player via the utility function.

\medskip

Players may also randomize their actions through the use of a \emph{mixed strategy}: a probability distribution $s_i$ over $i$'s pure strategies. Let $S_i = \Delta(\Rho_i)$ be the set of player $i$'s mixed strategies (where $\Delta(X)$ denotes the set of probability distributions over a domain $X$), and let $S=\bigtimes_{i \in N} S_i$ be the set of mixed strategy profiles.We overload the definition of utility function to accept mixed strategies as follows:
\begin{align*}
    u_i(s) = \sum_{\rho \in \Rho} \left(\prod_{i \in N} s_i(\rho_i) \right) u_i(\rho)
\end{align*}

In a two player game, when one player plays a pure strategy while the other plays a mixed strategy, we write 
\begin{align*}
    u_i(\rho_i, s_{-i}) = \sum_{\rho_{-i} \in \Rho_{-i}}  s_{-i}(\rho_{-i}) u_i(\rho_i, \rho_{-i})
\end{align*}

Where $-i$ is the other player from player $i$ in a two-player game. 
\begin{definition}
    A two-player game $G = (\{i, -i\}, \Rho, u) $ is  zero-sum\footnote{Note that two-player zero-sum games are no less general than 
 two-player constant-sum games, since we can turn any two-player constant-sum game into a strategically identical zero-sum game by either subtracting or adding a constant to each player's utility.} if 
    \begin{align*}
        u_i(\rho) + u_{-i}(\rho) = 0 \quad \forall \rho \in \Rho
    \end{align*} 
\end{definition}

\begin{definition}[$\epsilon$-Nash Equilibrium]
    A strategy profile $s$ is an  $\epsilon$-Nash equilibrium if $\forall i \in N$, we have 
    \begin{align*}
        u_i(s'_i, s_{-i}) -   u_i(s_i, s_{-i})  \leq \epsilon \quad \forall  s'_i \in S_i
    \end{align*}
\end{definition}

\begin{definition}[$\epsilon$-CCE \citep{Moulin-1978}]
    We say $\mu \in \Delta(\Rho)$ is an $\epsilon$-CCE if $\forall i \in N$, we have
    \begin{align*}
        \mathbb{E}_{\rho \sim \mu} \left[ u_i(s_i', \rho_{-i})-  u_i(\rho)  \right] \leq \epsilon \quad \forall s_i' \in S_i 
    \end{align*}
\end{definition}

\begin{definition}[Marginal strategy]
    Given  $\mu \in \Delta(\Rho)$, let $s_i^\mu$ be the \emph{marginal strategy} for $i$, where $s_i^\mu(\rho_i) \doteq \sum_{\rho_{-i} \in \Rho_{-i}} \mu(\rho_i, \rho_{-i})$. Let $s^\mu$ be a \emph{marginal strategy profile}, where each player plays their marginal strategy from $\mu$.
\end{definition}

\section{The Proof}
Here we show a simple proof that CCE imply Nash in two-player zero-sum games. We begin showing that if a player $i$ deviates from their CCE recommendations, if  $-i$ continues to play their CCE recommendations, this is equivalent to $-i$ playing their marginal strategy for the CCE.

\begin{lemma}\label{lemma:2.1}
      If $\mu \in \Delta(\Rho)$, then in a two-player game, for any deviation $s_i$, we have $\mathbb{E}_{\rho \sim \mu} \left[u_{i}(s_i, \rho_{-i}) \right]  = u_i(s_i, s^\mu_{-i})$ . 
\end{lemma}

\begin{proof}
    We have
    \begin{align*}
        \mathbb{E}_{\rho \sim \mu} \left[u_{i}(s_i, \rho_{-i}) \right] &= \sum_{\rho_i \in \Rho_i} \sum_{\rho_{-i} \in \Rho_{-i}} s_i(\rho_i) \mu_{-i}(\rho_{-i}) u_i(\rho_i, \rho_{-i}) ,
    \end{align*} 
    where $\mu_{-i}(\rho_{-i})  \doteq \sum_{\rho_{-i} \in \Rho_{-i}} \mu(\rho_i, \rho_{-i})$. Note, however, that $\mu_{-i}(\rho_{-i}) = s^\mu_{-i}(\rho_{-i}) $ so
    \begin{align*}
         \mathbb{E}_{\rho \sim \mu} \left[u_{i}(s_i, \rho_{-i}) \right] &= \sum_{\rho_i \in \Rho_i} \sum_{\rho_{-i} \in \Rho_{-i}} s_i(\rho_i) s^\mu_{-i}(\rho_{-i}) u_i(\rho_i, \rho_{-i})
        \\
        &= u_{i}(s_i, s^\mu_{-i}).
    \end{align*} 
\end{proof}

\begin{proposition}\label{prop:2.1}
     If $\mu$ is an $\epsilon$-CCE of a two-player constant-sum game, then $ |\mathbb{E}_{\rho \sim \mu} \left[ u_i(\rho) \right] -   u_i(s^\mu) | \leq \epsilon$
\end{proposition} 

\begin{proof}

    Suppose not, then $ \left | \mathbb{E}_{\rho \sim \mu} \left[ u_i(\rho) \right] -   u_i(s^\mu) \right| > \epsilon$ which means either
    \begin{align}
         u_i(s^\mu) - \mathbb{E}_{\rho \sim \mu} \left[ u_i(\rho) \right] > \epsilon \label{eq:1}
      \end{align}
    or
    \begin{align}
         \mathbb{E}_{\rho \sim \mu} \left[ u_i(\rho) \right] -   u_i(s^\mu)  > \epsilon \label{eq:2}.
    \end{align}
    Consider \eqref{eq:1}. By Lemma~\ref{lemma:2.1} we have $u_i(s^\mu) =  \mathbb{E}_{\rho \sim \mu} \left[ u_i(s_i^\mu, \rho_{-i}) \right]$, which means
    \begin{align*}
        \mathbb{E}_{\rho \sim \mu} \left[ u_i(s_i^\mu, \rho_{-i}) \right]-  \mathbb{E}_{\rho \sim \mu} \left[ u_i(\rho) \right]     > \epsilon,
    \end{align*}

    which contradicts the fact that $\mu$ is an $\epsilon$-CCE, since $s^\mu_i$ is a deviation that is more than $\epsilon$-profitable for player $i$. Next, consider \eqref{eq:2}; since the game is two-player zero-sum, we have:
    \begin{align*}
        & \mathbb{E}_{\rho \sim \mu} \left[ u_i(\rho) \right] -   u_i(s^\mu)  > \epsilon
        \\
        \implies & -\mathbb{E}_{\rho \sim \mu} \left[ u_{-i}(\rho) \right] - (-  u_{-i}(s^\mu))  > \epsilon
        \\
        \implies & u_{-i}(s^\mu)  - \mathbb{E}_{\rho \sim \mu} \left[u_{-i}(\rho) \right] > \epsilon.
    \end{align*}
    At this point we may repeat the steps above to show that $\mu$ is not an $\epsilon$-CCE, since $s^\mu_{-i}$ is a deviation that is more than $\epsilon$-profitable for player $-i$.
\end{proof}

\begin{proposition}
    If $\mu$ is an $\epsilon$-CCE of a two-player constant-sum game $G$, then $s^\mu$ is a $2\epsilon$-Nash equilibrium.
\end{proposition}
\begin{proof}
    Choose player $i$ arbitrarily. Either $u_i(s^\mu) \geq  \mathbb{E}_{\rho \sim \mu}\left[u_i(\rho)\right]$ or $u_i(s^\mu) <  \mathbb{E}_{\rho \sim \mu}\left[u_i(\rho)\right]$. Consider the first case. Starting from the definition of $\epsilon$-CCE:
    \begin{align*}
         & \mathbb{E}_{\rho \sim \mu} \left[u_i(\rho_i', \rho_{-i}) \right]  -  \mathbb{E}_{\rho \sim \mu} \left[u_i(\rho) \right]  \leq \epsilon   \quad \forall \rho_i' \in \Rho_i .
    \end{align*}
    But since $u_i(s^\mu) \geq  \mathbb{E}_{\rho \sim \mu}\left[u_i(\rho)\right]$ we have
    \begin{align*}
         & \mathbb{E}_{\rho \sim \mu} \left[u_i(\rho_i', \rho_{-i}) \right]  - u_i(s^\mu)  \leq \epsilon   \quad \forall \rho_i' \in \Rho_i .
    \end{align*}
    Which by Lemma~\ref{lemma:2.1} means
    \begin{align*}
        u_i(\rho_i', s^\mu_{-i}) -  u_i(s^\mu) \leq  \epsilon  \quad \forall \rho_i' \in \Rho_i .
    \end{align*}

    Thus $s^\mu$ is an $\epsilon$-Nash. Next, suppose $u_i(s^\mu) <  \mathbb{E}_{\rho \sim \mu}\left[u_i(\rho)\right]$. By Proposition \ref{prop:2.1} we have,  
    \begin{align}
        & \mathbb{E}_{\rho \sim \mu} \left[u_i(\rho) \right] - u_i(s^\mu) \leq \epsilon 
        \\
        \implies & \mathbb{E}_{\rho \sim \mu} \left[u_i(\rho) \right]   \leq    u_i(s^\mu) + \epsilon \label{eqn:ineq}.
    \end{align}
    Then, starting from the definition of $\epsilon$-CCE and applying \eqref{eqn:ineq},
    \begin{align*}
         & \mathbb{E}_{\rho \sim \mu} \left[u_i(\rho_i', \rho_{-i}) \right]  -  \mathbb{E}_{\rho \sim \mu} \left[u_i(\rho) \right]  \leq \epsilon   \quad \forall \rho_i' \in \Rho_i 
         \\
         \implies &  \mathbb{E}_{\rho \sim \mu} \left[u_i(\rho_i', \rho_{-i}) \right]  - ( u_i(s^\mu)  + \epsilon) \leq \epsilon   \quad \forall \rho_i' \in \Rho_i 
         \\
        \implies &  \mathbb{E}_{\rho \sim \mu} \left[u_i(\rho_i', \rho_{-i}) \right]  -    u_i(s^\mu) \leq 2 \epsilon   \quad \forall \rho_i' \in \Rho_i .
    \end{align*}
    By Lemma~\ref{lemma:2.1} we have 
    \begin{align*}
        u_i(\rho_i', s^\mu_{-i}) -  u_i(s^\mu) \leq  2\epsilon  \quad \forall \rho_i' \in \Rho_i .
    \end{align*}
\end{proof}

\end{flushleft}
\bibliographystyle{plainnat}
\bibliography{main.bib}

\begin{thebibliography}{3}
\providecommand{\natexlab}[1]{#1}
\providecommand{\url}[1]{\texttt{#1}}
\expandafter\ifx\csname urlstyle\endcsname\relax
  \providecommand{\doi}[1]{doi: #1}\else
  \providecommand{\doi}{doi: \begingroup \urlstyle{rm}\Url}\fi

\bibitem[Fudenberg and Tirole(1991)]{Fudenberg-1991}
Drew Fudenberg and Jean Tirole.
\newblock \emph{Game Theory}.
\newblock MIT Press Books. The MIT Press, 1991.

\bibitem[Greenwald and Jafari(2003)]{Greenwald-2003}
Amy Greenwald and Amir Jafari.
\newblock A general class of no-regret learning algorithms and game-theoretic
  equilibria.
\newblock In \emph{COLT}, volume~3, pages 2--12, 2003.

\bibitem[Moulin and Vial(1978)]{Moulin-1978}
H.~Moulin and J.P. Vial.
\newblock Strategically zero-sum games: The class of games whose completely
  mixed equilibria cannot be improved upon.
\newblock \emph{International Journal of Game Theory}, 7\penalty0 (3):\penalty0
  201--221, 1978.

\end{thebibliography}
\end{document}